\documentclass[12pt]{article}
\usepackage{graphicx}
\usepackage{amsmath}
\usepackage{multirow}
\usepackage{amssymb}
\usepackage{amsthm}
\usepackage{natbib}
\usepackage{moreverb}
\usepackage[T1]{fontenc}
\usepackage{gensymb}
\usepackage{hyperref}
\usepackage{pdfpages}
\usepackage[figuresright]{rotating}
\usepackage{geometry}
\usepackage{bbm}
\usepackage{setspace}
\usepackage{csquotes}
\parskip \bigskipamount
\renewcommand{\P}{{\mathbb{P}}}

\newcommand{\R}{{\mathbb{R}}}
\newcommand{\x}{{\mathbf{x}}}
\newcommand{\bl}{{\boldsymbol\ell}}

\newcommand{\bdelta}{{\boldsymbol\delta}}

\newcommand{\bu}{{\mathbf{u}}}

\newcommand{\X}{{\mathbf{X}}}

\newcommand{\var}{{\text{var}}}

\newcommand{\1}{\mathbbm{1}}

\newcommand{\z}{{\mathbf{z}}}
\newcommand{\Z}{{\mathbf{Z}}}

\newcommand{\cX}{{\mathcal{X}}}

\newcommand{\cF}{{\mathcal{F}}}
\newcommand{\cD}{{\mathcal{D}}}

\newcommand{\deltahat}{{\hat{\delta}}}

\newtheorem{theorem}{Theorem}
\begin{document}
%\maketitle
\begin{center}
\Large{\textbf{Discrete Optimization for Interpretable Study Populations and Randomization Inference in an Observational Study of Severe Sepsis Mortality}}\\
%\vspace{.3 in}
%\large{\textbf{Colin B. Fogarty$^{1*}$, Mark E. Mikkelsen$^{2,3}$, David F. Gaieski$^4$,\\ and Dylan S. Small$^{1}$}}\\
%\vspace{.3 in}
%\normalsize{$^{1}$Department of Statistics, The Wharton School, University of Pennsylvania,
%Philadelphia, PA 19104\\
%$^2$ Pulmonary, Allergy, and Critical Care Division, Perelman School of Medicine, University of Pennsylvania,
%Philadelphia, PA 19104\\
%$^3$Center for Clinical Epidemiology and Biostatistics, Perelman School of Medicine, University of Pennsylvania,  Philadelphia, PA 19104 \\
%$^4$ Department of Emergency Medicine, Thomas Jefferson University, Philadelphia PA 19107}\\
%\vspace{.1 in}
%\normalsize{*email: \texttt{cfogarty@wharton.upenn.edu}}
\end{center}
\abstract{Motivated by an observational study of the effect of hospital ward versus intensive care unit admission on severe sepsis mortality, we develop methods to address two common problems in observational studies: (1) when there is a lack of covariate overlap between the treated and control groups, how to define an interpretable study population wherein inference can be conducted without extrapolating with respect to important variables; and (2) how to use randomization inference to form confidence intervals for the average treatment effect with binary outcomes. Our solution to problem (1) incorporates existing suggestions in the literature while yielding a study population that is easily understood in terms of the covariates themselves, and can be solved using an efficient branch-and-bound algorithm. We address problem (2) by solving a linear integer program to utilize the worst case variance of the average treatment effect among values for unobserved potential outcomes that are compatible with the null hypothesis.  Our analysis finds no evidence for a difference between the sixty day mortality rates if all individuals were admitted to the ICU and if all patients were admitted to the hospital ward among less severely ill patients and among patients with cryptic septic shock. We implement our methodology in \texttt{R}, providing scripts in the supplementary material.}

\vspace{.15 in}
\newpage
\doublespace
\section{Introduction}\label{sec:intro}
\subsection{Severe Sepsis Incidence and Mortality}
Severe sepsis is a leading cause of morbidity and mortality worldwide. It is defined as a systematic inflammatory response to infection that is accompanied by acute organ dysfunction. \citet{ang01} estimate that severe sepsis afflicts roughly 750,000 individuals in the United States per year, of whom an estimated 215,000 perish. \citet{gai13} note that cases of severe sepsis appear to be on the rise. In a recent study, \citet{liu14} found that sepsis contributed to one in every two to three deaths in two complementary hospital cohorts, and suggest that ``improved treatment of sepsis (potentially a final hospital pathway for multiple other underlying conditions) could offer meaningful improvements in population mortality.''

A critical decision along this pathway is whether to admit a patient to an intensive care unit (ICU), or rather to an appropriate hospital ward. It is estimated that approximately 50 percent of severe sepsis patients in the United States are admitted to an ICU after presentation to an emergency department, with the rest being admitted to a hospital ward \citep{ang13}. Recent evidence suggests that admission to a non-ICU setting may be increasing \citep{whi14}. Severe sepsis varies in degree of gravity at time of presentation to the emergency department. In general, sicker patients tend to be placed in the ICU, and those exhibiting less severe symptoms are often admitted to the hospital ward. Furthermore, \citet{bru96} and \citet{roh13} note that there are systematic ways in which the epidemiology, site of infection, and organ dysfunctions appear to vary between ICU and hospital ward patients. 

The existing literature offers contrasting opinions on the optimal process of care for severe sepsis patients. \citet{est07} argue that there is a large population of patients not admitted to the ICU who could ``potentially benefit from more aggressive resuscitation and innovative therapies'' that are available in the ICU. They found that severe sepsis patients in hospital wards had a higher estimated mortality rate than those who were admitted to the ICU, although their result was not statistically significant. On the other hand, \citet{lev08} found that admission to an ICU covered by intensivists may result in worse health outcomes, in part because patients may receive unnecessary (but potentially harmful) therapies or procedures. It is feasible, then, that certain severe sepsis patients may be better off if they were admitted to the hospital ward, as they would not be subjected to interventions in the ICU that are not warranted given their condition. In keeping with this hypothesis, \citet{sun05} found that severe sepsis mortality rates among non-ICU patients were lower than those among ICU patients.

The goal of our analysis is to assess the causal effect of ICU admission versus hospital ward admission on health outcomes. To be precise, we aim to compare the average health outcomes if all individuals were admitted to the ICU with the average outcomes if all patients were admitted to the hospital ward.   We use data from a retrospective observational cohort study wherein hospital admissions of individuals with severe sepsis to the Hospital of the University of Pennsylvania between January 2005 and December 2009 were examined; see \citet{whi14} for further details on the data set. We only consider patients without hemodynamic septic shock (a patient has hemodynamic septic shock if the patient has severe sepsis coupled with hypotension after initial fluid resuscitation) because patients with hemodynamic septic shock are almost exclusively admitted to the ICU  (ProCESS Trial, 2014).  Investigators identified 1507 remaining individuals with severe sepsis but not hemodynamic septic shock, of whom 695 were admitted to an ICU and 812 were admitted to a hospital ward. Thirty covariates detailing demographic information, comorbidities, emergency department process of care, and site of infection were identified by expert consultation as germane to the hospital pathway and to health outcomes. We separated our covariates into three tiers of importance based on an \textit{a priori} assessment (i.e. before examining the data set) of their effect on admission decisions and mortality. Our health outcome is a binary variable that takes on the value 1 if a patient died any time between the date of admission and 60 days after hospital admission. The tier 1 covariates are listed in Table \ref{tab:covariates} along with their means and standard deviations among ICU and hospital ward patients, while remaining covariates are summarized in Appendix A.

\begin{table}
\begin{center}
\caption{Covariate Means and Standard Deviations, Original Population and Study Population for Tier 1 Covariates. The first two columns are the covariate means (standard deviations) in the initial study population, and the last two columns are the covariate means (standard deviations) in the study population defined in Section \ref{sec:subpop}. }
\vspace{.1 in}
\begin{tabular}{l c c c c}
\multicolumn{1}{l}{}&\multicolumn{2}{c}{Original Population}&\multicolumn{2}{c}{Study Population} \\
\hline
%\multicolumn{1}{|l}{Covariate}	& \multicolumn{1}{c}{Tier}	&\multicolumn{1}{c}{ICU Mean}	&\multicolumn{1}{c}{Ward Mean}&\multicolumn{1}{c}{ICU Mean}&\multicolumn{1}{c|}{Ward Mean}	\\
Covariate&ICU&Ward&ICU&Ward\\
\hline
\hline
Age 	&		60.1	&	55.1	&	60.56	&	55.88	\\
	&			(17.4) &	(18.4)	&	(17.1)	&	(18.3)	\\
Charlson comorbity index 	&		2.52	&	2.41	&	2.43	&	2.48	\\
	&			(2.81)	&	(2.64)	&	(2.70)	&	(2.65)\\
Initial serum lactate 	&		4.26	&	2.56	&	3.22	&	2.61	\\
	&			(2.98)	&	(1.23)	&	(1.24)	&	(0.956)	\\
APACHE II score 	&		17.7	&	13.6	&	16.9	&	13.8	\\
	&			(6.37)	&	(5.27)	&	(5.46)	&	(4.73)	\\
\hline
\end{tabular}
\label{tab:covariates}
\end{center}
\end{table}

%Cryptic septic shock	&	 Exact 	&	44\%	&	10\%	&	31\%&	9\%	\\
A subgroup of severe sepsis patients who are of particular interest to the critical care community are those with \textit{cryptic septic shock}. These are severe sepsis patients who have normal levels of systolic blood pressure (so do not have hemodynamic septic shock) yet exhibit high levels of initial serum lactate ($\geq 4$ mmol/L) \citep{pus11}. Initial serum lactate levels refer to the amount of lactic acid in the blood upon presentation to an emergency department. Initial serum lactate levels have been associated with mortality for severe sepsis patients independent of organ dysfunction, and are therefore thought to be a highly useful biomarker for risk-stratifying patients upon presentation to an emergency department \citep{mik09}. Some believe that cryptic septic shock patients should be classified as septic shock patients and admitted to an ICU by default, while others suggest that there may be no benefit to such a protocol; see \citet{jon11} and \citet{riv11} for both sides of the debate. Hence, in addition to comparing ICU versus hospital ward mortality among all severe sepsis patients without hemodynamic septic shock, we would further like to compare mortality within the subgroup of cryptic septic shock patients, as this subgroup may exhibit mortality outcomes that differ from other severe sepsis patients.  While only 10\% of patients admitted to the hospital wards had cryptic septic shock in our sample, this number was 44\% for patients admitted to the ICU.

\subsection{From Observational Study to Idealized Experiment}
Randomization inference provides an appealing framework even when the data are not the result of a randomized experiment. This is in keeping with the advice of  H.F. Dorn, as relayed in \citet{coc65}, that \enquote{the planner of an observational study should always ask himself the question, \enquote{how would the study be conducted if it were possible to do it by controlled experimentation?}} Through matching on observed covariates, we attempt to mimic a well-balanced randomized experiment. Matching methods encourage researcher blinding, since matched sets can and should be constructed without looking at the outcome of interest. Using randomizations within this idealized experiment as the basis for inference also allows us to assess the robustness of a study's finding to unmeasured confounding through a sensitivity analysis. See \citet{obs} for a discussion of using randomization inference within observational studies.

Towards this end, we employ covariate matching to account for measured confounders that may bias our comparison of 60 day mortality rates if all patients had been admitted to the ICU versus if all patients had been admitted to the hospital ward, and then conduct inference with respect to the match that is produced; see \citet{stu10} for a comprehensive overview of common matching algorithms. Full matching, the algorithm used herein, is a type of matching algorithm that optimally assigns individuals into strata consisting of either one treated unit and many control units or one control unit and many treated units, and is particularly appealing for studies where the ratio of treated individuals to control individuals is close to 1:1.  See \citet{ros91} and \citet{han04} for additional details on full matching. 

In Section \ref{sec:mimic}, we discuss the randomized experiment that full matching aims to replicate. We begin our analysis in Section \ref{sec:firstattempt}, where we discuss an issue encountered within our comparison of hospital wards and ICU that is common to many observational studies: an inherent lack of covariate overlap.  In Section \ref{sec:maxbox}, we discuss how the \textit{maximal box problem} marries together existing methods for addressing lack of covariate overlap with the intuitive appeal of a study population whose boundaries are clearly defined in terms of important covariates. 
 
Section \ref{sec:randinf} lays out the necessary framework for conducting inference on the \textit{average treatment effect} in the idealized experiment we aim to uncover. Difficulties arise due to the composite nature of a null hypothesis on the average treatment effect, in that different allocations of potential outcomes can yield the same average treatment effect while inducing different randomization distributions for its estimate. We overcome these difficulties by finding a sharp upper bound on the variance of the estimated average treatment effect over all elements of the composite null, which under a normal approximation allows us to carry out inference for the composite null in question. In Section \ref{sec:inference}, we apply our methodology to our sepsis example. 

Though seemingly unrelated, our solutions for defining an interpretable study population and conducting randomization inference on the average treatment effect with binary outcomes utilize methods from discrete optimization. Traditionally, discrete optimization problems were viewed as tractable if the worst case instance could be solved by an algorithm that grows polynomially in the instance's size, and statistician have typically limited themselves to using algorithms of this type. Both of the problems we pose are $\mathcal{NP}$-hard, meaning that there is no known polynomial time algorithm for the worst case instances of these problems.  However, there have been recent advances in solving typical cases of these problems such that a typical case of these problems can often be solved in a reasonable amount of time \citep{sch03}. In a recent paper, \citet{zub12} highlighted the usefulness of mixed integer programming for attaining well balanced matched strata. We illustrate that when applying the methods described in this paper to our data set, solutions can be attained in a matter of seconds. Through the methods developed in this work, we hope to further emphasize the usefulness of discrete optimization for observational studies and statistics in general.

\section{Review of Causal Inference via Matching}\label{sec:mimic}

\subsection{Notation For a Stratified Randomized Experiment} \label{sec:rand}
Suppose there are $I$ total strata, the $i^{th}$ of which contains $n_i \geq 2$ individuals. In each stratum, $m_i \geq 1$ individuals receive the treatment, $n_i = m_i - n_i$ individuals receive the control, and $\min\{m_i, n_i-m_i\} = 1$. Furthermore, $m_i$ is fixed across randomizations, resulting in $n_i$ distinct assignments to treatment and control for each stratum $i$. Assignments are independent between distinct strata. Under the potential outcomes framework with binary responses, each individual has two potential binary outcomes: one under treatment, $r_{Tij}$, and one under control, $r_{Cij}$, which are 1 if an event would occur and 0 otherwise. The true treatment effect for individual $j$ in stratum $i$ is $\delta_{ij}  = r_{Tij} - r_{Cij}$, and is unobservable as each individual receives either treatment or control. The observed response for each individual is $R_{ij} = r_{Tij}Z_{ij} +r_{Cij}(1-Z_{ij})$, where $Z_{ij}$ is an indicator variable that takes the value 1 if individual $j$ in stratum $i$ is assigned to the treatment; see, for example, \citet{ney23} and \citet{rub74}. Each individual has observed covariates $\x_{ij}$.

There are $N = \sum_{i=1}^I n_i$ individuals in the study, of whom $N_T = \sum_{i=1}^I m_i$ receive the treatment and $N_C = N - N_T$ receive the control. Let $\mathbf{R} = (R_{11}, R_{12},...,$ $R_{I, n_{I}})^T$ and $\mathbf{Z} = (Z_{11}, Z_{12}, ..., Z_{I, n_{I}})^T$. Let $\Omega$ be the set of $\prod_{i=1}^I n_i$ possible values $\z$ of $\mathbf{Z}$ under the given stratification. In a randomized experiment, randomness is modeled through the assignment vector; each $\z \in \Omega$ has probability $1/|\Omega|$ of being selected. Hence, quantities dependent on the assignment vector such as $\Z$ and $\mathbf{R}$ are random, whereas $r_{Tij}$, $r_{Cij}$, $\x_{ij}$ are fixed quantities. Let $\cF = \{r_{Tij}, r_{Cij}, \x_{ij}, i = 1,..,I, j = 1,...,n_i\}$. For a randomized experiment, we can then write that $\P(Z_{ij} = 1 | \cF, \Z \in \Omega) = m_i/n_i$,  $i=1,..,I; j = 1,...,n_i$ and that $\P(\Z = \z| \mathcal{F}, \Z \in \Omega) = 1/|\Omega|$. 

\subsection{Matching and Observational Studies}\label{sec:assume}
In an observational study, we begin with an unmatched study population of size $N$.
Matching methods aim to create strata where the constituent individuals have similar covariate values, or at a minimum similar probabilities of assignment to treatment \citep{ros83, stu10}. Once a match is obtained, the acceptability of the resulting stratification is assessed for covariate balance through the use of various diagnostics, the most common of these being the standardized difference \citep{designofobs}. Let the notation introduced in Section \ref{sec:rand} now apply to the stratification yielded by the matching algorithm. If the match passes the balance diagnostics, randomization inference then proceeds under the assumptions of no unmeasured confounding, common support for the assignment probabilities, and equal probabilities of assignment within a matched strata. The assumption of no unmeasured confounding states that given the observed covariates, the probabilities of assignment to treatment are independent of the potential outcomes, that is $\P(Z_{ij}=1|\x_{ij}) = \P(Z_{ij}=1|\x_{ij}, r_{Tij}, r_{Cij})$, $i=1,...,I; j=1,...,n_i$. This probability is known as the \textit{propensity score}, and we denote it by $e(\x_{ij})$. The assumption of common support for the assignment probabilities can be written as $ 0 < e(\x_{ij}) < 1$, $i=1,...,I; j=1,...,n_i$. Finally, the assumption of equal probability of treatment assignment within a matched strata can be written as $e(\x_{ij}) = e(\x_{ik})$ for all $i = 1,...,I$; $j, k = 1, . . . , n_i$. Under these assumptions, we have that $\P(Z_{ij} = 1 | \cF, \Z \in \Omega) = m_i/n_i$, $i=1,..,I; j = 1,...,n_i$ and that $\P(\Z = \z| \mathcal{F},\Z \in \Omega) = 1/|\Omega|$, thus recovering the randomized experiment described in Section \ref{sec:rand}. 

\section{Lack of Common Support}

\subsection{Imbalance Caused by Limited Covariate Overlap}\label{sec:firstattempt}
We begin by conducting a full match on our entire study population. As was previously noted, we have 30 pre-treatment covariates that were deemed important for both the probability of admission to the ICU versus the ward and for the outcome. Of these, 13 contained missing values; see Appendix A for the percentages of missing observations for these 13 covariates. To account for this, we include 13 new missingness indicators, and fill in the missing values with the mean of the covariates. As is discussed in \citet{ros84} and \citet[][Section 9.4]{designofobs}, this facilitates balancing both the observed covariates and the pattern of missingness between the two groups being compared. We also include an indicator for whether an individual has cryptic septic shock. We thus have 44 covariates that could be used in constructing our matched sets. In determining which variables to match on, the avoidance of various types of ``collider-bias''  \citep{gre03}  must be considered. We first do not control for any post-treatment variables in order to avoid biases that stem from controlling for the consequence of an exposure. One particular type of collider bias, $M$-bias, can be induced even when only controlling for pre-treatment variables. Despite this, we choose to control for all 44 of these pre-treatment covariates because of the work of \citet{din14mbias}, simulation studies of \citet{liu12}, and arguments of \citet{rub09} that suggest that biases stemming from not controlling for a relevant pre-treatment covariate tend to be more substantial than those that are caused by $M$-bias.

We use a rank-based Mahalanobis distance with a propensity score caliper of 0.2 standard deviations as our distance metric between ICU and hospital ward patients, where the propensity scores are estimated via a logistic regression of our covariates on the treatment indicator; for further discussion on the role of propensity score calipers in multivariate matching, see \citet[][Section 8.3]{designofobs}. In addition, we match exactly on the cryptic septic shock indicator, meaning that each stratum produced by the full match must either contain all cryptic septic shock patients or none. We use standardized differences, defined as a weighted difference in means divided by the pooled standard deviation between groups before matching, to assess balance in our resulting matched strata for the remaining covariates \citep{stu08}. A common rule of thumb is to deem the balance of a resulting match acceptable if all absolute standardized differences fall below 0.1 \citep{designofobs}. We modify this rule slightly based on our covariate importance tiers, using thresholds of 0.05, 0.10, and 0.15 for the standardized differences of tiers 1, 2, and 3 respectively. Thus, we require more stringent balance for those covariates that are deemed to be of highest importance for the admission decision and for mortality.

We first perform an unrestricted full matching. Without any restrictions, full matching can produce extremely large strata.  When applied to our data set, there are strata with ratios of hospital ward patients to ICU patients of 37:1, 1:21, 1:32, and 1:65. Noting the potential for outlandishly large strata, \citet{han04} advocates placing a bound on the maximal allowable strata size in order to increase the effective sample size (and thus, the power of the resulting analysis). In keeping with this, we also performed full matches with restricted ratios of hospital ward patients to ICU patients within a stratum, with ratios ranging from 2:1, 1:2 to 15:1, 1:15. Neither the unrestricted full match nor any of the restricted full matches resulted in an adequately balanced matched sample based on our standardized difference thresholds.

Our failure to attain a suitably balanced stratification does not suggest a deficiency with full matching; to the contrary, no matching algorithm should be able to produce a suitably balanced stratification without discarding individuals, as there is a severe lack of covariate overlap between patients admitted to the ICU and patients admitted to the hospital wards. Two covariates that were out of balance in all of the restricted ratio matches were initial serum lactate levels and APACHE II scores. As is described in Section \ref{sec:intro}, initial serum lactate is believed to be important for both the admission decision and for health outcomes. The APACHE II score is a measure of disease severity using physiologic variables and chronic health conditions \citep{kna85}.  As Figure \ref{fig:maxbox} displays, virtually all of the patients admitted to the hospital ward lie in the lower left hand quadrant of the scatterplot of APACHE II scores versus initial serum lactate levels. Naturally, this lack of overlap arises because many ICU patients are more severely ill than any hospital ward patient. We cannot possibly infer the effect of admission to the ICU versus the hospital ward on mortality for the severely ill ICU patients, as we lack patients admitted to the hospital wards with which the outcomes of those ICU patients can be fairly compared. Assessment of causal effects for those individuals would represent an analysis of  ``extreme counterfactuals,'' resulting in an extrapolation to which the data cannot honestly attest \citep{kin06}. Rather, inference about the effect of being admitted to an ICU or a hospital ward on mortality must be restricted to the area of common support (i.e., those patients who were less gravely ill at presentation), a fact to which restricted ratio full matches bear testament in their inability to attain suitable balance.

\begin{figure}[h]
\begin{center}
\includegraphics[scale = .5]{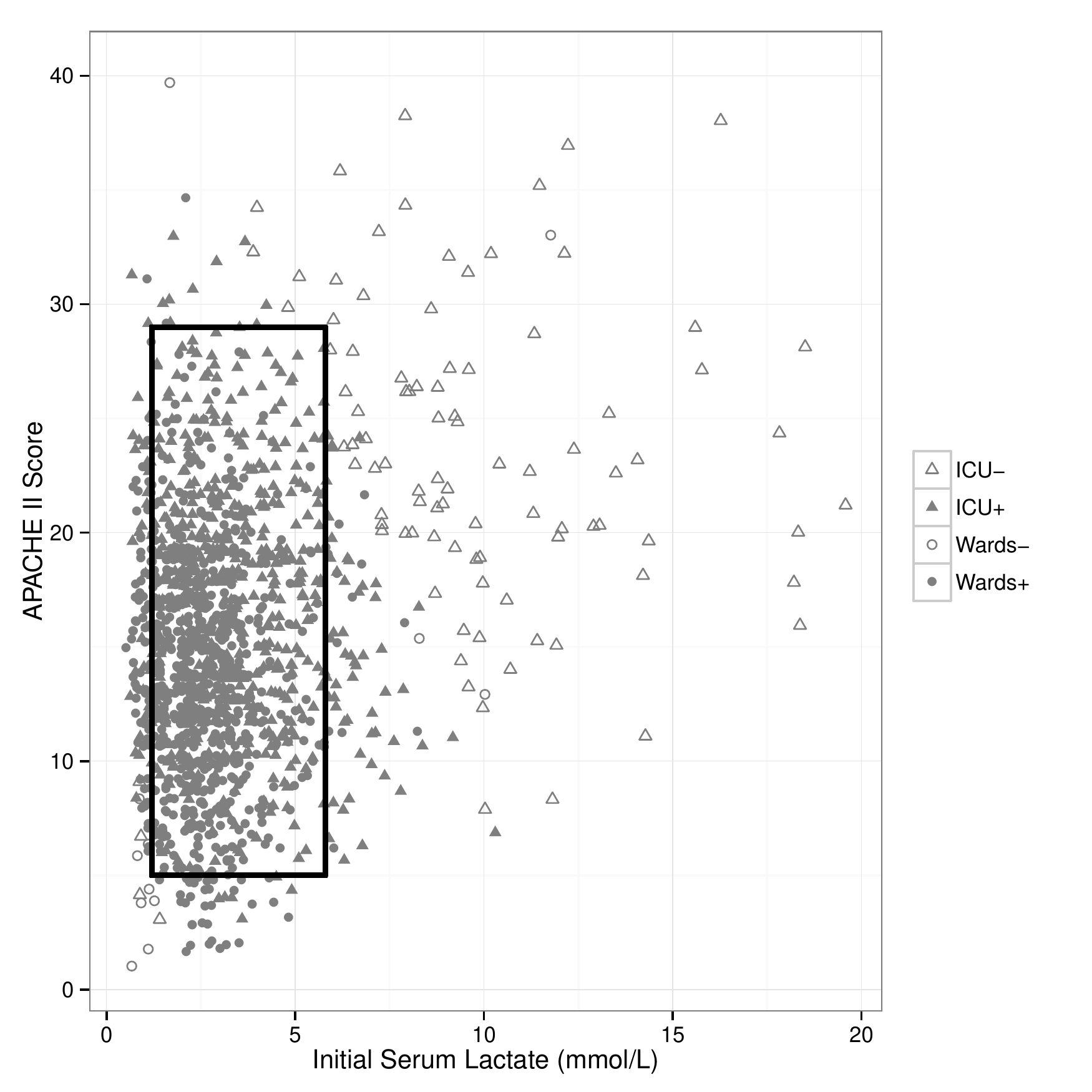}
\caption{\small{Lack of Common Support and the Maximal Box. This figure shows a scatter plot of initial serum lactate levels and APACHE II scores. The plot also shows the \textit{maximal box}, which is the solution to the optimization problem posed in Section \ref{sec:maxbox}. The rectangular boundaries represent the study population identified as having a common support, wherein subsequent inference will be restricted. It was formed by finding the rectangle containing the largest number of filled points, subject to excluding all hollow points in the plot. The triangles represent ICU patients, and the circles represent hospital ward patients.  Whether a point is filled or hollow is described in Section \ref{sec:subpop} in detail, and has to do with whether or not it was determined that a given individual was in the area of viable common support for his or her observed tier 1 covariates.  Points are jittered to avoid overplotting.}}
\label{fig:maxbox}
\end{center}
\end{figure}

\subsection{Different Types of Overlap}

Before proceeding, we discuss a few different notions of covariate overlap.  The first notion, which we call strong overlap, is that for every treated unit in the data, there is a control unit that has similar or the same covariate values and that for every control unit, there is a treated unit that has similar or the same covariate values.  While strong overlap is most desirable and can be readily diagnosed in low dimensions through visual tools such as scatterplots, it is difficult to obtain when there are a moderate or high number of covariates because of the curse of dimensionality.  The second notion, which we call interpolation overlap, is that for any treated unit, an estimate of that treated unit's counterfactual control potential outcome given the unit's covariates can be inferred through an interpolation rather than an extrapolation of the observed control outcomes and that for any control unit, an estimate of the control's unit counterfactual potential outcome can be inferred through interpolation.  \citet{kin06}  present an operational way to check for interpolation overlap  by means of the convex hull of the treated and control covariate distributions. According to their criterion, one is performing interpolation if a given treated (control) individual is in the convex hull of the control (treated) covariate distributions, and is performing extrapolation otherwise. Interpolation overlap then exists if all treated units are in the convex hull of the control units, and all control units are in the convex hull of the treated units. Unfortunately, as noted in \citet{kin06} their interpolation overlap criterion is also difficult to obtain in moderate and high dimensions. In Appendix B, we demonstrate through a simulation study that even when the treated and control covariate distributions are identical, the number of individuals for which ``interpolation'' is identified as being performed by the convex hull diagnostic decreases substantially as the covariate dimension increases.

\subsection{Existing Methods for Achieving Overlap}
A lack of overlap is typically addressed by defining a study population restriction wherein adequate overlap can be attained. Many methods are motivated by the fact that, asymptotically, strong overlap is present if and only if the propensity score at a given covariate value, $e(\x_j)$,  is bounded away from 0 and 1 for all individuals $j \in \{1,...N\}$. In this sense, the propensity score provides a scalar indication of both the existence of and the extent of covariate overlap.  \citet{deh99} recommend removing treated units whose propensity scores are larger than the maximal propensity score among the control units, and removing control units whose propensity score are smaller than the minimal propensity score among the treated units. \citet{cru09} define a study population by seeking the subset of the covariate space which minimizes the efficiency bound for the variance of the study population average treatment effect.  Based on this optimality criteria, they find that for a wide range of distributions a close approximation to the optimal selection rule is to drop all units with estimated propensity scores outside of [0.1, 0.9]. One concern with propensity score approaches for attaining overlap for finite sample inference is that while boundedness away from 0 and 1 implies strong overlap asymptotically, for finite samples treated (control) individuals with nonzero propensity scores may still lack comparable control (treated) individuals in terms of their observed covariates. Another concern is that these propensity scores must be estimated, so that individuals with nonzero estimated propensity scores may nonetheless fall outside the area of overlap.

Other methods directly deal with the covariates themselves when defining a new study population. \citet{kin06} identify a multivariate space wherein one performs interpolation rather than extrapolation by removing treated individuals whose covariates lie outside of the convex hull of the covariates for the control individuals, and removing control individuals whose covariates lie outside of the convex hull of the covariates for the treated individuals. \citet{ros12} describes a method for optimal subsampling wherein one chooses an upper bound on how many treated units can be removed from the resulting matched sample. \citet{hil13} employ Bayesian Additive Regression Trees \citep{chi10} to identify areas of common support, using the fact that the variability of individual-level conditional expectations tend to increase drastically in such areas. Individuals are then classified as being inside or outside the area of common support based on thresholds for these variances.

Though easy to implement and often accompanied by theoretical justifications, the resulting study population returned by these methods is often unappealing as it may be difficult to interpret in terms of the covariates themselves. This makes it difficult to succinctly and transparently describe the individuals to whom the performed inference applies. Furthermore, for study populations defined by propensity scores alone, a researcher's notion of which individuals have high or low ``propensity'' for treatment may be vastly different from the individuals designated as such through fitting a propensity score model to the data. A practitioner not participating in the study could then have a misconception of the individuals to whom the inference applies based on his or her preconceived notion of which individuals are likely to receive treatment or control.  In his \textit{Design of Observational Studies} book, Rosenbaum advises that when excluding extreme individuals ``it is usually better to go back to the covariates themselves, $\x_{j}$, perhaps redefining the population  under study to be a subpopulation of the original population'' \citep[Section 3.3.3]{designofobs}. \citet{stu10} further echoes this sentiment, arguing that ``it can help the interpretation of results if it is possible to define the discard rule using one or two covariates'' \citep[page 15]{stu10}.  

To illustrate the potential confusion arising from a study population definition in terms of propensity scores, suppose we decided to apply the suggestion of \citet{cru09} to our tier 1 covariates in order to define our study population. In its most succinct form, the resulting study population would be defined as $\{i: \text{logit}(3.5 - 0.0049(\text{age}_i) + 0.069(\text{CCI}_i) - 0.46(\text{init. ser. lac}_i) - 0.12(\text{APACHE II}_i)) \in [0.1, 0.9]\}$. The boundaries of this set would likely hold little meaning to practitioners, as it is hard to characterize \textit{qualitatively} the individuals who fall within these bounds. Inference performed on this subset would pertain to a set of individuals who lack a clear characterization on the basis of the covariates of interest themselves, limiting how actionable the findings may be. 

\citet{tra11} suggest a tree based approach for defining an internally valid study population based on values of covariates alone. In the first step of their method, the practitioner uses a pre-existing methods for study population definition of her choice; any of those described at the beginning of this section would be valid choices. For each individual, this outputs an indicator of whether or not that individual belongs to the area of common support (and hence, should be included in the new study population). The user next fits a regression tree of a designated depth that aims to minimize the probability of misclassification, and defines the study population based on the resulting tree (rather than by the method used in the initial step). While resulting in a markedly more interpretable study population, by their very nature trees result in interval restrictions that are path dependent, rather than intervals that are universally applicable for all individuals.

Restrictions to rectangular regions of the covariate space are appealing as they can be explicitly defined in terms of the intersection of a series of intervals, rather than as a complicated function of the observed covariates.  Each interval pertains to a unique covariate, allowing one to paint a coherent description of the resulting study population  through covariate-specific constraints. This allows the practitioner to clearly understand the restriction that each covariate imposes on the study population. Currently, little guidance exists on how to define these covariate based inclusion criterion. Ad-hoc choices based on inspection may discard large proportions of individuals, and further may fail to discard individuals who are identified as problematic.

\subsection{An Attainable Objective}
As outlined in this section, there are inherent difficulties with attaining strong overlap in high dimensions. We thus instead seek to define a study population characterized by three principles which are both attainable and verifiable. Firstly, we would like the study population to demonstrate overlap with respect to those covariates deemed most important for the treatment and the outcome. By limiting ourselves to a small set of important covariates to focus on for overlap, both strong and interpolation overlap can be potentially obtained for a reasonably sized study population. Furthermore, the overlap with respect to these most important covariates can be verified using visual aids such as scatterplots. Secondly, our study population should be such that balance can be attained on all covariates. As balance is a property of the marginal distributions for the treated and control individuals, standard metrics such as standardized differences can speak to balance being attained for all covariates. Finally, we would like our study population to have a simple definition in terms of important covariates while not being overly wasteful in discarding individuals.

Our approach to achieving these goals is two-fold. We begin by constructing, through the solution to the \textit{maximal box problem}, a study population that incorporates existing methods for identifying individuals outside the area of common support with respect to important covariates, retains as many viable individuals as possible, and is readily interpretable based on important covariates as it defined through the intersection of interval restrictions. After this, we use full matching to arrive upon a stratification that mimics a well-balanced randomized experiment within this study population. We then proceed with inference in the resulting study population only if the balance on all covariates is deemed acceptable.

\section{Defining a Study Population}\label{sec:maxbox}
\subsection{The Maximal Box Problem}
A box $[\bl, \bu]$ is defined to be a closed interval (hyperrectangle) of $\R^p$,
\begin{align*} [\boldsymbol\ell, \bu] &:= \{\x \in \R^p: \ell_i\leq x_i \leq u_i \; \forall i \in \{1,..,p\}\} \end{align*}

Suppose one has a finite collection of vectors $\{\x_{j}\}, j = 1,...,N,$ that can be partitioned into two disjoint sets of ``positive'' points, $\cX^+$ and ``negative'' points, $\cX^-$. The maximal box problem aims to find the lower and upper boundaries of a box, $[\tilde{\bl}, \tilde{\bu}]$, such that the corresponding box contains the maximal number of points in $\cX^+$ while containing none of the points in $\cX^-$. Explicitly, $[\tilde{\bl}, \tilde{\bu}]$ is the $\arg\max$ of the following optimization problem (MB, for maximal box):
\begin{align}
\tag{MB} 
\label{eq:MB}\text{maximize } & \;\; |[\bl, \bu] \cap \cX^+|\\
\text{subject to } &\; \; |[\bl, \bu] \cap \cX^-| = 0,\nonumber
\end{align}where the notation $|A|$ denotes the number of elements of set $A$. Henceforth, we refer to $|[\bl, \bu] \cap \cX^+|$ as the \textit{cardinality} of a box $[\bl, \bu]$.

\citet{eck02} describe the problem in detail. They prove that the problem is $\mathcal{NP}$-hard in general, but is polynomial time for any fixed dimension $p$. They provide an efficient branch and bound algorithm for solving it, which they show to have modest computation time in practice. They also provide a mixed integer programming formulation of the problem, which facilitates its use with freely available and commercial solvers.

\subsection{From Maximal Boxes to Study Populations}
Let $\mathbf{D}(\x_{j}, \X, \Z)$ be a binary decision rule that determines whether a point $\x_{j}$ needs to be excluded from the analysis to ensure covariate overlap. For example, the recommendations of \citet{deh99}, denoted $\mathbf{D}_{DW}(\x_{j}, \X, \Z)$, and the rule proposed in \citet{cru09} (the simplified version of the rule), denoted as $\mathbf{D}_{C}(\x_{j}, \X, \Z)$, can be written in this form as:
\begin{align*}\mathbf{D}_{DW}(\x_{j}, \X, \Z) &= \begin{cases} \1\left\{\hat{e}(\x_{j}) \leq \max\{\hat{e}(\x_{k}) \text{ s.t. } Z_{k} = 0\} \right\}& \text{if  } Z_{j} = 1\\
 \1\left\{\hat{e}(\x_{j}) \geq \min\{\hat{e}(\x_{k}) \text{ s.t. } Z_{k} = 1\} \right\}& \text{if  } Z_{j} = 0\\
\end{cases}\\
\mathbf{D}_{C}(\x_{j}, \X, \Z) &= \1\left\{\hat{e}(\x_{j}) \in [0.1, 0.9]\right\}
\end{align*}

Our sets of positive and negative points are then defined based on the decision rule, with $\cX^+ := \{\x_{j}: \mathbf{D}(\x_{j}, \X, \Z)  = 1\}$, and $\cX^ -:= \{\x_{j}: \mathbf{D}(\x_{j}, \X, \Z)  = 0\}$. We then solve (\ref{eq:MB}) using these designations of positive and negative points. The resulting maximal box is one that contains the largest possible number of observations who could feasibly have been in the study population, while eliminating all individuals who were designated for exclusion. The study population defined by the maximal box has a clear interpretation in terms of the covariates themselves: an individual is in the study population if $\tilde{\bl} \leq \x_{j} \leq \tilde{\bu}$, and is excluded otherwise.

We note that as $p$ (the number of  covariates used to define the maximal box) increases, the number of positive points in corresponding maximal box is non-decreasing. At the same time, this increases the potential computational burden, as there are at most $|\mathcal{X}^+|^{2p}$ possible candidates for the boundaries of the maximal box \citep{eck02}. Thus, in practice we recommend forming the boundaries on the maximal box based on values of the most important covariates. Note that defining a study population on the basis of important covariates can also be justified on the basis of interpretability. If one defined a study population using a maximal box formed from a large number of covariates, the resulting study population would likely be just as cryptic as one determined solely by the estimated propensity scores. Further, \citet{hil13} argue that methods for common support restriction should primarily consider those covariates that are most important for the outcome. As such, we seek to define a study population based on the most important pre-treatment covariates. We also recommend using covariates that are not binary for constructing the maximal box as the resulting restriction would either eliminate one of the categories entirely, or (more commonly) be the whole range [0,1]. If there is a binary covariate of considerable importance, we recommend accounting for it by either exactly matching or almost exactly matching on the binary covariate \citep[][Sections 9.1 and 9.2]{designofobs} for details.

There is a possibility that the resulting maximal box only contains a small fraction of the positive points. This means there is no easy way to define a region of good overlap between the treated and control individuals without eliminating the vast majority of the data. In Appendix C, we discuss an extension of the maximal box problem posed in \citet{eck02} that may be appropriate in this setting. This generalization allows for a small number, $C$, of points marked for exclusion (negative points) to be included within the bounds of the maximal box, which would in turn allow for the incorporation of more positive points; see Appendix C for more discussion on the ramifications of choosing $C>0$. In our example, we proceed with $C=0$, thus requiring the exclusion of all points marked as being outside the area of viable support.

\subsection{Application to Our Original Population}\label{sec:subpop}
As defining a maximal box with all 44 covariates would yield a highly unwieldy 44 dimensional box with limited interpretability, we instead aim to construct a maximal box using our four tier 1 covariates: age, Charlson comorbidity index, APACHE II scores, and initial serum lactate levels. Our approach is to fit a propensity score model using a logistic regression on our four tier 1 covariates, and to then employ the simplified criterion of \citet{cru09} with these propensity scores to determine which observations had to be removed. We use this reduced propensity score model because individuals within the area of common support on our important variables may be nonetheless extreme with respect to other, less important, covariates, which may in turn lead to them being marked for removal if we used the full propensity score model. As our focus is on attaining covariate overlap \textit{and} balance for our most important variables while seeking balance on all other variables, we wanted our exclusion metric to reflect lying in the area of covariate overlap with respect to our most important variables. See Appendix D for a more detailed discussion of this goal and the behavior of alternative strategies. Denoting the tier 1 covariate for individual $j$ as $\x^{(1)}_{j}$, our decision rule is $\mathbf{D}_{C, \text{Tier1}}(\x_{j}, \X, \Z) = \1\left\{\hat{e}\left(\x^{(1)}_{j}\right) \in [0.1, 0.9]\right\}$. This results in  108 individuals being marked for exclusion. We have implemented the branch and bound algorithm of \citet{eck02} in the \texttt{R} programming language \citep{R14}, and used it to find our study population; a script for our implementation is provided in the supplementary materials. For this data set, our implementation took 2 seconds to run on a desktop computer with a 3.40 GHz processor and 16.0 GB RAM. 

We created a maximal box using all four tier 1 covariates, and also created one using only initial serum lactate and APACHE II scores. The cardinalities of these boxes were very close to one another (1214 and 1208 respectively). As such, we use the box defined using only initial serum lactate and APACHE II scores for enhanced interpretability. The resulting maximal box is displayed as the rectangle in Figure \ref{fig:maxbox}. As can be seen, the study population under investigation can be explicitly defined as those individuals in our initial study whose APACHE II scores are between 5 and 29 and whose initial serum lactate levels are between 1.2 and 5.8 mmol/L. Our study population thus restricts analysis to those individuals who had less severe, but not the least severe, conditions upon presentation to the emergency department. The study population defined by the maximal box includes 701 out of 812 patients admitted into the wards and 507 out of 695 patients admitted to the ICU, resulting in 1208 out of the original 1507 individuals being available for further study; furthermore, it contains 86.3\% of all individuals whose estimated propensity scores were deemed acceptable by our decision rule. Table \ref{tab:covariates} shows the means and standard deviations of the tier 1 covariates among this study population; values for the other covariates can be found in Appendix A. As can be seen, restricting ourselves to this study population improved pre-matching balance for many of the covariates.

We now proceed with a full matching on our study population of 1208 individuals whose condition upon presentation was less severe. We first refit our propensity score model on this study population to exploit the so-called balancing property of the propensity score within our population of interest \citep{ros83}. We use a rank-based Mahalanobis distance with a propensity score caliper of 0.2 standard deviations computed with respect to this study population alone to define distance between ICU and hospital ward patients. Further, we require exact matching for the cryptic septic shock indicator. Given our distances, we run a series of full matches ranging from most restrictive to less restrictive until a suitably balanced matched sample could be attained. We found that a 1:7, 7:1 restricted full matching was able to adhere to the standardized difference tolerances defined in Section \ref{sec:firstattempt}, as is displayed in Figure \ref{fig:balancesub}. 

\begin{figure}[h]
\begin{center}
\includegraphics[scale = .5]{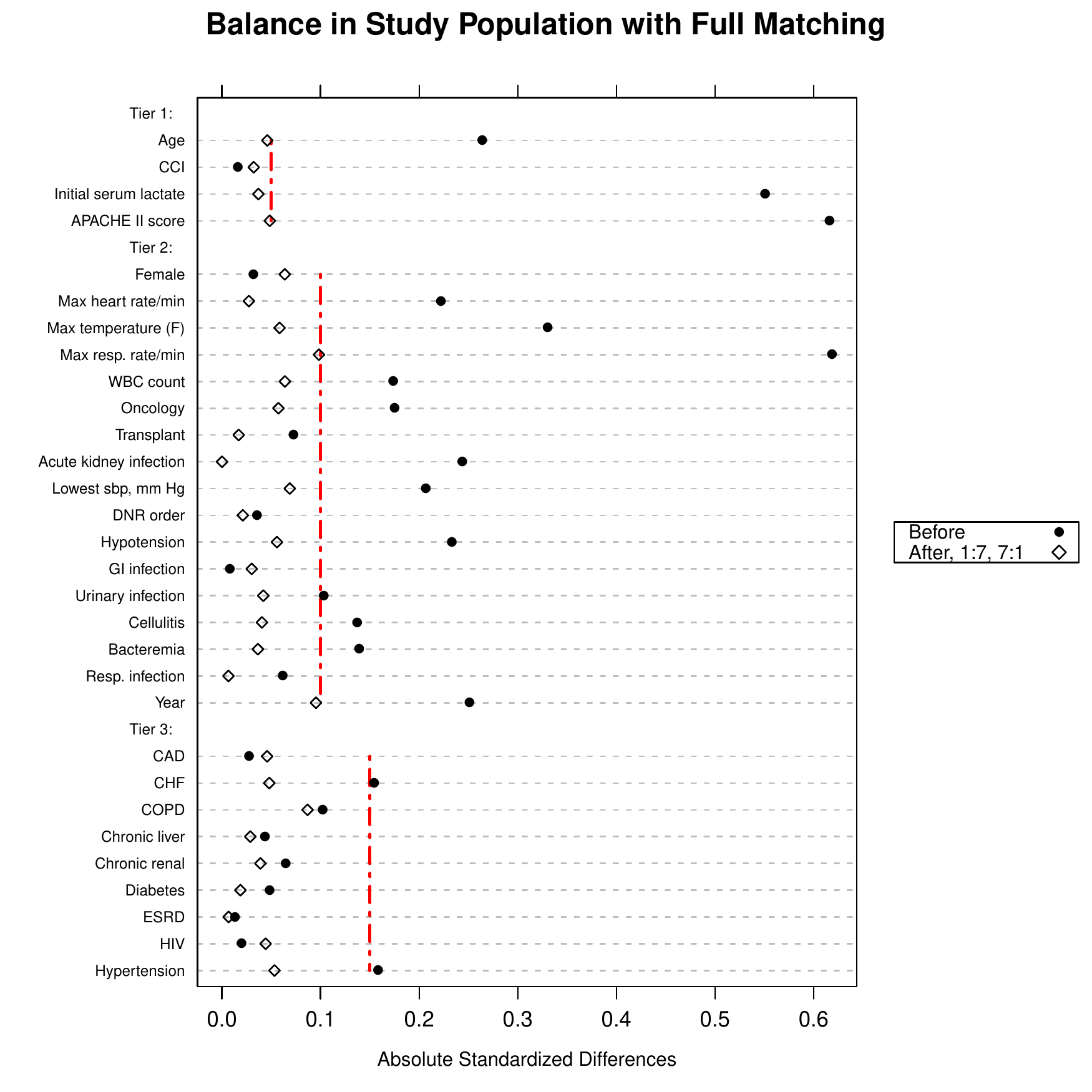}
\caption{\small{Covariate Imbalances Before and After Full Matching, Study Population. The dotplot (a Love plot) shows the absolute standardized differences without matching, and after conducting a restricted 1:7, 7:1 full matching on our study population. The vertical dotted lines correspond to the standardized difference tolerances for each of the three covariate tiers. Although not shown here, all standardized differences corresponding to our 13 missingness indicators had standardized differences below 0.1, indicating that the pattern of missing data was also balanced between the ICU and hospital ward groups.}}
\label{fig:balancesub}
\end{center}
\end{figure}

With our match deemed suitably balanced, we are ready to proceed with inference within the stratified randomized experiment we have aimed to replicate. Our goal is to assess not only whether there is a substantial difference in mortality rates depending on admission to the ICU or the hospital ward, but also to measure the extent of the effect. In order to do so, we now discuss performing inference and constructing confidence intervals for the average treatment effect.

\section{Randomization Inference for the Average Treatment Effect with Binary Outcomes}\label{sec:randinf} \label{sec:ate}
The average treatment effect with binary outcomes (also known as the causal risk difference) is the difference between the proportion of positive responses among the potential outcomes under treatment and the potential outcomes under control, $\delta := (1/N)\sum_{i=1}^I \sum_{j=1}^{n_i}\delta_{ij}$. It it identifiable under the assumption of strong ignorability \citep{ros83}, and an unbiased estimator of the average treatment effect under a stratified design is given by  $\hat{\delta} := \sum_{i=1}^I(n_i/N)\hat{\delta}_i$, where $\hat{\delta}_i =\sum_{j=1}^{n_{i}}\left( Z_{ij}R_{ij}/{m_i} - (1-Z_{ij})R_{ij}/(n_i - m_i)\right)$ is the estimated average treatment effect within stratum $i$ \citep[Section 2.5]{obs}. 

We consider tests of the null hypothesis that $(1/N)\sum_{i=1}^I \sum_{j=1}^{n_i}\delta_{ij} = \delta_0$, $\delta_0 \in \{d/N: d \in [-N, N]\cap \mathbb{Z}\}$, where $\mathbb{Z}$ denotes the set of all integers. In reality, a null hypothesis of this form is a large collection of hypotheses on the set of treatment effects, $\bdelta = [\delta_{11}, \delta_{12}, ..., \delta_{I, n_{i}}]$. Let $\mathcal{D}_{\delta_0}$ be the set of all $\bdelta$ such that $(1/N)\sum_{i=1}^I \sum_{j=1}^{n_i}\delta_{ij} = \delta_0$ and such that the treatment effects are compatible with the observed data. The latter requirement means that if unit $j$ in stratum $i$ received the treatment in the observed experiment, the value of $r_{Tij}$ is fixed at $R_{ij}$ and hence $\delta_{ij}$ can equal either $R_{ij}$ or $R_{ij}-1$. If said unit received the control, the value of $r_{Cij}$ is fixed at $R_{ij}$, and $\delta_{ij}$ can equal either $-R_{ij}$ or $1-R_{ij}$. To reject a null hypothesis $(1/N)\sum_{i=1}^I \sum_{j=1}^{n_i}\delta_{ij} = \delta_0$, we require that we reject the null hypothesis that the allocation of treatment effects equals $\bdelta$ for all $\bdelta \in \mathcal{D}_{\delta_0}$.

\subsection{Existing Methods}
\label{sec:ney}
Inspired by the work of \citet{ney23}, randomization inference for average treatment effect across randomization is typically conducted by finding a consistent estimator of an upper bound on the variance of the estimated ATE resulting in randomization inference that asymptotically has the proper Type I error rate; see \citet{din14} among many. \citet{rob88} improves upon the upper bound of \citet{ney23} for binary outcomes under an unstratified design and uses the resulting upper bound to create confidence intervals that are narrower than those based on a Wald-type procedure. More recently, \citet{aro14} provide asymptotically sharp upper bounds on $\var(\hat{\delta})$ under general potential outcomes.

For a stratified design, the variance for the estimated ATE is
\begin{align}\label{eq:ubstrat}\var(\hat{\delta}) &=  \sum_{i=1}^I \frac{n_i^2}{N^2} \left(\frac{S_{Ti}^2}{m_i} + \frac{S_{Ci}^2}{n_i - m_i} - \frac{S_{\delta i}^2}{n_i}\right) \end{align}
where $S_{Ti}^2 = \sum_{j=1}^{n_i}(r_{Tij} - \bar{r}_{Ti})^2/(n_i - 1)$,  $S_{Ci}^2 = \sum_{j=1}^{n_i}(r_{Cij} - \bar{r}_{Ci})^2/(n_i - 1)$, and $S_{\delta i}^2 = \sum_{j=1}^{n_i}(\delta_{ij} - \bar{\delta}_{i})^2/(n_i- 1)$. The procedures of \citet{ney23}, \citet{rob88} and \citet{aro14} can be readily extended to stratified designs where $m_i$ and $n_i - m_i$ are sufficiently large for each stratum $i$. However, these procedures have deficiencies when there are strata for which either $m_i$ or $n_i- m_i = 1$, as these procedures require an estimate of the variance of the treated and control groups in each strata. When $m_i$ or $n_i-m_i = 1$, unbiased estimators for $S_{Ti}^2$ or $S_{Ci}^2$ do not exist. Matched strata returned by pair matching, fixed ratio matching, variable ratio matching and full matching have this property, rendering the existing bounding techniques based solely on in-sample estimates inapplicable.

\citet{rig14} present two methods for conducting randomization inference and constructing confidence intervals for the average treatment effect with binary outcomes in an unstratified design. The first method proceeds by combining two tests on the \textit{attributable effect} of \citet{ros01} and \citet{ros02att} through means of a Bonferroni correction. They then mention that this approach, while potentially conservative, can be readily applied to stratified randomized experiments. In the second method, hypothesis testing proceeds by conducting randomization inference on $\delta$ for all $\bdelta \in \cD_{\delta_{0}}$, meaning that this procedure has level $\alpha$ for testing the corresponding composite null. Confidence intervals are then constructed by inverting tests for values under the null $\delta_0 \in \{d/N: d \in [-N, N]\cap \mathbb{Z}\}$, where $\mathbb{Z}$ again denotes the set of all integers. In their description, inference is conducted by explicitly performing a randomization test for each $\bdelta \in \cD_{\delta_{0}}$.  Noting the inherent computational burden in this process as $N$ increases in an unstratified experiment, they suggest a Monte-Carlo procedure to approximate the required permutation test. For stratified experiments, they suggest that this approach becomes computationally unwieldy quite quickly, thus advocating the use of a potentially conservative method based on the attributable effect in this setting.

Our procedure combines elements of the classical Neyman approach and the hypothesis test inversion approach of \citet{rig14}. Our approach is not purely Neymanian in that although we are testing Neyman's null hypothesis, we do not proceed by seeking a consistent upper bound on $\var(\deltahat)$; rather,  we explicitly compute the largest value of $\var(\deltahat)$ possible among the elements of $\cD_{\delta_{0}}$ for each null hypothesis. The resulting bound on the variance of the average treatment effect for a given null hypothesis is sharp, as it is attained by a member of the composite null $\bdelta^* \in \cD_{\delta_{0}}$. As a test of a composite null hypothesis is size $\alpha$ only if the supremum over all elements of the composite null of the probability of rejection is $\alpha$, asymptotically our testing procedure has exactly size $\alpha$ asymptotically so long as a normal approximation is justified. This is because since the numerator is the same for the test statistic for any null in $\mathcal{D}_{\delta_0}$, namely $\hat{\delta}-\delta_0$, the p-value computed under a normal approximation will be maximized by the member of the composite null with the largest denominator of the test statistic, i.e, the member with the largest variance. Rejection on the basis of this worst-case $p$-value then implies rejection for all elements of the composite null. For finite samples, discrepancies in actual versus advertised size stem only from the strength of the normal approximation. We show in Appendix E that for our case study, the true distribution of the average treatment corresponding to the worst-case allocations of potential outcomes is well approximated by a normal distribution. In Appendix F, we discuss why our standard errors are necessarily larger than those attained in other common scenarios (for example, in testing Fisher's sharp null).

As will be discussed in Section \ref{sec:opt}, the use of a normal approximation allows us to overcome the computational issues encountered in \citet{rig14}. This normal approximation can be justified under very mild conditions. Let $\sigma^2_i = n_i^2(S_{Ti}^2/m_i+ S_{Ci}^2/(n_i - m_i) - S_{\delta i}^2/{n_i})$ be the contribution to $\var(N\hat{\delta})$ from strata $i$ (i.e., $\sum_{i=1}^I\sigma_i^2/N^2 = \var(\deltahat))$, and let $\Sigma = \sum_{i=1}^I\sigma^2_i$. Let $n^*$ be an upper bound on the maximal size of a stratum. 

\begin{theorem}{If $\Sigma \rightarrow \infty$ as $I\rightarrow \infty$, then $(N\hat{\delta} - N\delta)/\sqrt{\Sigma} \overset{d}{\rightarrow} \mathcal{N}(0, 1)$.} \end{theorem}
\begin{proof}
Since our outcomes are binary, the maximal contribution of an individual summand $ n_i\hat{\delta}_i$ to $\sum_{i=1}^In_i\hat{\delta}_i = N\hat{\delta}$ is bounded in absolute value by $n^*$. Using Lyapunov's central limit theorem applied to a sequence of independent bounded random variables \citep[Corrolary 2.7.1]{leh04}, we have that $(N\hat{\delta} - N\delta) \overset{d}{\rightarrow} \mathcal{N}(0, \Sigma)$ as $I\rightarrow \infty$ provided that $\Sigma\rightarrow\infty$ as $I\rightarrow\infty$.
\end{proof}

This requirement precludes a certain type of degeneracy. Namely, it cannot be the case that only finitely many strata have nonzero variances for $\hat{\delta_i}$. This, coupled with a bound on the maximal strata size, suffices for asymptotic normality to hold.

\subsection{Integer Programming for the Maximal Variance}\label{sec:opt}
In theory, the maximal variance for a given composite null, $H_0: \delta = \delta_0$, could be found by enumerating all $2^N$ possible allocations of unobserved binary potential outcomes, computing $\var(\hat{\delta})$ through (\ref{eq:ubstrat}) for each allocation, and finding the maximal variance among the allocations that satisfy $\bdelta \in \mathcal{D}_{\delta_{0}}$. Such a na\"{\i}ve approach quickly becomes computationally infeasible: in our application, this would would require enumerating $2^{1208}$ sets of potential outcomes. 

Our approach is to instead pose the problem of maximizing the variance within a composite null as an integer program. Roughly stated, the resulting integer program optimizes the variance over the values of the unobserved potential outcomes, subject to the resulting allocation of potential outcomes being a member of the composite null. Though many equivalent formulations of the desired optimization problem are possible, the one we choose explicitly avoids symmetric solutions, known to cripple the computation time of integer programs \citep{mar10}, by having each decision variable correspond to a unique distribution on the contribution to the overall estimated average treatment effect from a given stratum. Our approach exploits three essential facts. Firstly, there is symmetry \textit{between} strata in that (a) $n_i = n_{i'}$, (b) $\sum_{j=1}^{n_i}Z_{ij}R_{ij} = \sum_{j=1}^{n_i'}Z_{i'j}R_{i'j}$ and (c) $\sum_{j=1}^{n_i}(1-Z_{ij})R_{ij} = \sum_{j=1}^{n_i'}(1-Z_{i'j})R_{i'j}$ for symmetric strata $i$ and $i'$, meaning that any allocation of potential outcomes for strata $i$ is also a feasible allocation for strata $i'$. Secondly, there is symmetry \textit{within} strata in that $Z_{ij} = Z_{ik}$ and $R_{ij}=R_{ik}$ for symmetric individuals $j$ and $k$ in stratum $i$, meaning that the $\var(\hat{\delta}_i)$ remains the same if the values for the unobserved potential outcome are permuted among symmetric individuals in stratum $i$. Finally, there is independence between strata which allows us to sum strata-wise variance contributions together to arrive at the overall variance of the estimated average treatment effect. In combination, these three facts allow this seemingly daunting optimization problem to be solved in a matter of seconds. See Appendix G for a detailed discussion of our integer programming formulation.

\section{Inference for Severe Sepsis Mortality}\label{sec:inference}
We now proceed with randomization inference on the study population defined by our maximal box in Section \ref{sec:subpop}. As a reminder, this consists of severe sepsis patients without hemodynamic septic shock, with initial serum lactate between 1.2 and 5.8 mmol/L, and with APACHE II scores between 5 and 29. Of the 1208 patients in our study population, 701 were admitted to the hospital ward and 507 were admitted to the ICU. Our causal estimand is the difference between 60 day mortality rates if all patients had been admitted to the ICU and if all patients had been admitted to the hospital ward. Before matching, the unadjusted (and hence potentially biased) estimates for these rates under ICU and hospital ward admissions are 24.3\%  and 12.0\% respectively overall, and are 27.7\% and 21.2\% respectively within the cryptic septic shock subgroup.

%After adjusting for measured confounders through covariate matching, the estimated mortality rates under ICU and hospital ward admission are 19.4\% and 15.1\% respectively overall, and are 26\% and 26.8\%  respectively within the cryptic septic shock subgroup.  Table \ref{tab:results} shows the estimated average treatment effects (the differences between proportions under ICU and hospital ward admission) both in our overall study population and among the cryptic septic shock subgroup. We also report 95\% confidence intervals, which were formed by inverting a series of hypothesis tests as discussed in Section \ref{sec:opt}. Initializing the constants of the problem that do not vary with the null being tested took 1.4 seconds, and the integer program required for each hypothesis test took 0.1 seconds solve using \texttt{Gurobi} and 0.2 seconds to solve using  the freely available \texttt{lpSolve} on a desktop computer with a 3.40 GHz processor and 16.0 GB RAM. This demonstrates that confidence intervals can be constructed using our integer programming formulation efficiently using both commercial and freely available solvers.

After adjusting for measured confounders through covariate matching, the estimated mortality rates under ICU and hospital ward admission are 19.4\% and 15.1\% respectively overall, and are 26.0\% and 26.8\%  respectively within the cryptic septic shock subgroup.  Table \ref{tab:results} shows the estimated average treatment effects (the differences between proportions under ICU and hospital ward admission) both in our overall study population and among the cryptic septic shock subgroup. We also report 95\% confidence intervals, which were formed by inverting a series of hypothesis tests as discussed in Section \ref{sec:ney}. Both of these confidence intervals contain 0, indicating that we lack substantial evidence to suggest that there is a nonzero effect both overall and in the cryptic septic shock subgroup. Through our implementation, we were able to construct the reported confidence intervals in 0.42 seconds using \texttt{Gurobi}, and 0.72 seconds to solve using the freely available \texttt{lpSolve} on a desktop computer with a 3.40 GHz processor and 16.0 GB RAM. This demonstrates that confidence intervals can be constructed using our integer programming formulation efficiently using both commercial and freely available solvers. 

\begin{table}
\begin{center}
\caption{Estimated differences in severe sepsis mortality between patients admitted to the ICU and the hospital ward in our study population, both overall and among patients with cryptic septic shock. Positive values favor hospital ward admission, and negative values favor ICU admission. The standard errors reported are the Wald-estimates $\sqrt{\nu^T\x^{*}_{\hat{
\delta}}}$ and confidence intervals were constructed by inverting a series of tests as described in Section \ref{sec:ney}. }
\vspace{.1 in}
\begin{tabular}{c c c}
\multicolumn{1}{c}{} & \multicolumn{1}{c}{Overall} & \multicolumn{1}{c}{Cryptic S.S.} \\
\hline \hline
Estimated ATE & 4.3\%&-0.8\%  \\
(SE) & (3.7\%)  & (9.0\%)\\
\hline
95\% Conf. Int. & [-3.0\%; 11\%]  &[-18\%; 17\%] \\
\hline
\end{tabular}
\label{tab:results}
\end{center}
\end{table}

\section{Discussion}

As expected, we found that common support was not present for the most severely ill sepsis patients. The subset of septic shock patients, which include those with hemodynamic compromise or evidence of hypoperfusion, are routinely admitted directly to the ICU and therefore an observational study cannot address the effect of these triage decisions.For the population with substantial common support, our findings suggested that there was no clear benefit to direct ICU admission for non-shock, severe sepsis patients. In fact, recognizing our wide confidence intervals, the magnitude of the potential benefit of direct ICU admission after adjusting for all measured confounders through matching at the leftmost extreme of our confidence interval was relatively small at 3\%. While larger studies are required to substantiate our findings, our analysis suggests that the common practice in hospitals with strained ICUs (occupancy rates approaching 100\%) to defer ICU admission for many severe sepsis patients does not result in demonstrable harm to the patients.

By using the maximal box problem to define a study population for further analysis, we arrived at a study population with a readily interpretable interpretation in terms of important covariates wherein acceptable balance could be attained. One downside of our method is that it is not guaranteed that suitable balance can be attained in the returned study population. That is, one may arrive at a study population defined in terms of important covariates where it is difficult to find a matching procedure that attains suitable balance on all covariates. One option is to simply iterate: covariates for which suitable balance cannot be achieved can be used in defining a study population through the maximal box problem, and then one could again try to attain balance within the proposed study population. An interesting area for future research would be to create a procedure where the returned study population is guaranteed to have a match with acceptable balance. With fixed ratio matching, recent work on mixed-integer programming matching \citep{zub12} and cardinality-matching \citep{zub14} may provide insight into how to incorporate the balancing constraints into the optimization problem. 

In our application, we determined which covariates were most important for the treatment and the outcome (and hence those for which we seek verifiable overlap) through consultation with subject matter experts. In other applications, practitioners may not want to rely solely on prior information for determining which covariates are important and rather allow the data itself to attest to this. While model selection for the propensity score model can be conducted without concern, one must be careful when assessing the impact of covariates on the outcome variable as it could potentially bias the resulting inference by compromising the ``researcher blinding'' that makes matching so appealing \citep{rub06}. One path forward would be to employ sample splitting, thus assessing importance of covariates for the outcome using data that is not involved in the matched analysis.

Through our analysis of the impact of ward versus ICU admission on 60 day mortality rates, we have shown that the applicability of discrete optimization in causal inference extends far beyond matching algorithms. In fact, discrete optimization provides a powerful set of tools for solving many problems common to observational studies and, more broadly, statistics in general. The availability of efficient solvers can serve as the impetus for new methods that trade potentially unverifiable model assumptions for an increase in computation time. This is not to say that computational burden should not be considered when developing statistical methodology; rather, it is to caution against limiting the imagination solely on the basis of the computational power of the present day. As history has borne out, what is intractable today may be feasible tomorrow. 

\begin{center}
{\large\bf SUPPLEMENTARY MATERIAL}
\end{center}
\begin{description}
\item[Appendices] Appendix A provides summary statistics for all of the covariates used in matching. In addition, it contains the percentages of missingness for the 13 covariates with missing values. Appendix B demonstrates the difficulty of verifying interpolation overlap in moderate and high dimensions. Appendix C discusses an extension of the maximal box problem that allows for including up to $C$ negative points instead of 0. Appendix D compares the study population derived in Section 4.2 to those derived using other exclusion criteria. Appendix E displays the appropriateness of the normal approximation for inference on the ATE in this example. Appendix F discusses the value of the standard error used for conducting inference, and compares it to standard errors for three other procedures for conducting inference with binary outcomes. Appendix G contains the technical details of our integer programming formulation. (.pdf file)
\item[\texttt{R}-script for maximal boxes:] \texttt{maxbox.R}  provides code for producing a maximal box. (.R file)
\item[\texttt{R}-script for binary ATE:]  \texttt{ATEbinary.R} provides functions for estimation and inference on the ATE when responses are binary. One can compute the exact maximal standard errors using the \texttt{Gurobi} commercial solver or the freely available \texttt{lpSolve} package in \texttt{R}. Academic licenses for \texttt{Gurobi} are freely available, as is an \texttt{R} package. See their website for details. (.R file)
\end{description}

\bibliographystyle{apalike}
\bibliography{bibliography}

\end{document}